\documentclass[a4paper,10pt]{amsart}
\usepackage[utf8]{inputenc}
\usepackage{amsmath}
\usepackage{amsfonts}
\usepackage{amsthm}
\usepackage{amssymb}

\newcommand{\p}[1]{\mathbb{P}}
\newcommand{\fa}[1]{\forall}

\newcommand{\bP}{\mathbb{P}}

\newcommand{\ti}{\times}
\newcommand{\al}{\alpha}
\newcommand{\si}{\sigma}

\newcommand{\R}{\mathbb{R}}
\newcommand{\cO}{\mathcal{O}}
\newcommand{\cD}{\mathcal{D}}
\newcommand{\cP}{\mathcal{P}}

\theoremstyle{definition}
\newtheorem{defi}{Definition}[section]

\theoremstyle{remark}
\newtheorem{rem}[defi]{Remark}

\theoremstyle{plain}
\newtheorem{lemma}[defi]{Lemma}

\newtheorem{prop}[defi]{Proposition}
\newtheorem{thm}[defi]{Theorem}

\allowdisplaybreaks

\newcommand{\sgn}{\operatorname{sgn}}
\newcommand{\med}{\operatorname{median}}
\newcommand{\Var}{\operatorname{Var}}
\newcommand{\Span}{\operatorname{Span}}
\newcommand{\Sp}{\mathbb{S}}
\newcommand{\n}[1]{\left\lVert#1\right\rVert}
\newcommand{\eu}[1]{\left\lVert#1\right\rVert_2}
\newcommand{\skp}[2]{\langle #1, #2\rangle}
\newcommand{\Cov}[2]{\operatorname{Cov}(#1, #2)}
\newcommand{\eps}{\varepsilon}

\newcommand{\del}{\delta}
\newcommand{\Pb}{\mathbb{P}}
\newcommand{\E}{\mathbb{E}}

\newcommand{\od}{\odot}

\begin{document}

\title[Fast binary embeddings with Gaussian circulant matrices]{Fast binary embeddings with Gaussian circulant matrices: improved bounds}

\author{Sjoerd Dirksen}
\address{RWTH Aachen University, Lehrstuhl C f{\"u}r Mathematik (Analysis), Pontdriesch 10, 52062 Aachen,  Germany}
\email{dirksen@mathc.rwth-aachen.de}

\author{Alexander Stollenwerk}
\address{RWTH Aachen University, Lehrstuhl C f{\"u}r Mathematik (Analysis), Pontdriesch 10, 52062 Aachen,  Germany}
\email{stollenwerk@mathc.rwth-aachen.de}

\keywords{Binary embeddings, Johnson-Lindenstrauss embeddings, circulant matrices}
\subjclass[2010]{60B20,68Q87}

\maketitle

\begin{abstract}
We consider the problem of encoding a finite set of vectors into a small number of bits while approximately retaining information on the angular distances between the vectors. By deriving improved variance bounds related to binary Gaussian circulant embeddings, we largely fix a gap in the proof of the best known fast binary embedding method. Our bounds also show that well-spreadness assumptions on the data vectors, which were needed in earlier work on variance bounds, are unnecessary. In addition, we propose a new binary embedding with a faster running time on sparse data.  
\end{abstract}

\section{Introduction}

In this paper we are concerned with the problem of encoding a set of vectors $\cD$ located on the sphere in a high-dimensional space into a small number of bits. Our goal is to obtain an efficient encoding which approximately retains information on the angular distances between the points in the data set. Formally, we want to construct a \emph{$\del$-binary embedding} of $\cD$, i.e., a map $f:\Sp^{n-1}\rightarrow\{-1,1\}^m$ and a distance $d$ on $\{-1,1\}^m$ so that 
$$|d(f(p),f(q)) - d_{\Sp^{n-1}}(p,q)| \leq \del \qquad \text{for all } p,q \in \cD,$$
where $d_{\Sp^{n-1}}$ denotes the normalized geodesic distance on the sphere. Recent works have shown that this goal can be achieved using maps of the form
$$f_A(x) = \sgn(Ax), \qquad x\in\R^n,$$
with $A\in\R^{m\ti n}$. That is, one first embeds the data into a lower-dimensional space using a linear map and subsequently takes entry-wise signs of the embedded vectors. Apart from the intrinsic interest in binary embeddings for complexity reduction and computational purposes, they play an important role in one-bit compressed sensing \cite{JLB13,PlV13} and have been considered as computationally cheap layers in deep neural networks \cite{CYFKCC15}.\par 
In many successful approaches, $A$ is a random matrix drawn independently of the data. For instance, is not hard to show that if $A$ is an $m\ti n$ standard Gaussian matrix $G$ and $\cD\subset \Sp^{n-1}$ is any finite set of $N$ points, then with probability exceeding $1-\eta$, $f_A$ is a $\del$-binary embedding for $\cD$ into $(\{-1,1\}^m,d_H)$, provided that the number of bits satisfies $m\gtrsim \del^{-2}\log(N/\eta)$ \cite{PlV14,YCP15}. Here, $d_H$ denotes the normalized Hamming distance and $\lesssim$ hides an absolute constant. It is known \cite{YCP15} that this bit complexity is \emph{optimal}: any oblivious random map from $\cD$ into $\{-1,1\}^m$ which is a $\del$-binary embedding with probability at least $1-\eta$ must satisfy $m\gtrsim \del^{-2}\log(N/\eta)$.\par 
Although the Gaussian binary embedding achieves the optimal bit complexity on finite sets, it has a clear computational downside: as a Gaussian random matrix is densely populated, the implementation involves slow matrix-vector multiplications. A simple idea to decrease the running time, suggested in \cite{YCP15}, is to first reduce the dimensionality of the data by using a fast Johnson-Lindenstrauss transform \cite{AiC09} before applying a Gaussian binary embedding. We will refer to this map as the \emph{accelerated Gaussian binary embedding}. For appropriate choices of the involved parameters, this map $f_A$ with $A=G\Phi_{\operatorname{FJL}}$, is a $\del$-binary embedding into $(\{-1,1\}^m,d_H)$ with optimal bit complexity and, if $\log N\lesssim \del^{2}n^{1/2}$, runs in time $\cO(n\log n)$. However, outside of this parameter range the running time deteriorates.\par 
To obtain near-linear time and improved space complexity in a wider range of parameters, focus has recently shifted from purely Gaussian matrices to more structured random matrices. It was experimentally observed \cite{YKG14} that an equally good embedding into $(\{-1,1\}^m,d_H)$ can be obtained by replacing the dense Gaussian matrix $A$ by a subset of $m$ rows of a Gaussian circulant matrix with randomized column signs. That is, one considers  $A=R_I C_g D_{\eps}$, where, given $I\subset [n]$, $R_I$ is a subsampling operator that restricts a vector to its entries indexed by $I$, $C_g$ is a circulant matrix generated by a standard Gaussian vector and $D_{\eps}$ is a diagonal matrix with independent random signs on the diagonal. This matrix allows for matrix-vector multiplication in time $\cO(n\log n)$ by exploiting the fast Fourier transform. The experiments in \cite{YKG14} indicate that the performance of the embedding deteriorates if one leaves out the column sign randomization.\par 
In view of these experiments, it is desirable to try to prove that the matrix $A=R_I C_g D_{\eps}$ induces a binary embedding with the optimal bit complexity. As $d_H(f_A(p),f_A(q))$ is an unbiased estimator of the geodesic distance $d_{\Sp^{n-1}}(p,q)$, it would suffice to show that $d_H(f_A(p),f_A(q))$ concentrates strongly around its mean for fixed points $p,q$. Due to the non-linearity of the sign function and the dependencies between the circulant rows, this is a tough problem. In fact, even an optimal bound for the variance of $d_H(f_A(p),f_A(q))$ is not known. Nevertheless, some interesting partial results have recently been established, which can roughly be grouped in two directions. In both cases, the results have been used to construct new fast binary embeddings by appropriately modifying the matrix $A$ using additional randomness.\par 
In the first direction, developed by the authors in \cite{YBK15} and later quantitatively improved in \cite{Oym16}, it is shown that the map $f_A$ associated to $A=R_IC_gD_{\eps}$ is a $\del$-binary embedding on a finite set $\cD$ with large probability, provided that \emph{all vectors in $\cD$ as well as in $\cD-\cD$ are well-spread}. One can show that an arbitrary set of points will have these two properties with high probability after they have been pre-processed using a randomized Hadamard transform. To be more precise, the corresponding result of \cite{Oym16} is as follows. Let $I$ be a set of $m$ indices selected uniformly at random, let $g_1,g_2$ be independent standard Gaussian vectors, $D_{g_1}$, $D_{g_2}$ the diagonal matrices with $g_1,g_2$ on the diagonal and let $H$ be an $n\ti n$ Hadamard matrix. Consider $A=R_I C_g D_{g_1} H D_{g_2}$ and let $\cD\subset \Sp^{n-1}$ be any set of $N$ points. If 
$$\log N\lesssim \del^2 (\log n)^{-1} n^{1/3}, \qquad m\gtrsim \del^{-3}\log N,$$ 
then with high probability $f_A$ is a $\del$-binary embedding of $\cD$ into $(\{-1,1\}^m,d_H)$. In contrast to the accelerated Gaussian binary embedding, $f_A(x)$ can always be computed in time $\cO(n\log n)$. However, the parameter range in which $f_A$ is provably a $\del$-binary embedding is more restrictive.\par
The second direction, pursued by the authors in \cite{YCP15}, seeks to prove an \emph{optimal variance bound} for $d_H(f_A(p),f_A(q))$, meaning that the bound decays as $\cO(1/m)$. Once this is established, one can stack independent copies of $A$ into a matrix to obtain a high probability binary embedding. To be precise, \cite{YCP15} introduced the map $f_A$, with $A=\Psi\Phi_{\operatorname{FJL}}$, where $\Phi_{\operatorname{FJL}}$ is a $n'\ti n$ fast Johnson-Lindenstrauss transform and
$$\Psi=\begin{pmatrix}
\Psi^{(1)}\\
 \vdots\\
\Psi^{(B)}
\end{pmatrix}\in \R^{m\times n'}
,$$
consists of $B$ independent $m/B\times n'$ subsampled Gaussian Toeplitz matrices with randomized column signs. It was claimed in \cite{YCP15} that for appropriate choices of $n'$ and $B$, the map $f_A$ is a $\del$-binary embedding into $(\{-1,1\}^m, d_{\operatorname{med},B})$, where $d_{\operatorname{med},B}$ is the median over the $B$ block-wise Hamming distances. The embedding attains the optimal bit complexity and it runs in $\cO(n\log n)$ when $\log N\lesssim \del n^{1/2}/(\log(1/\del))^{1/2}$. To our knowledge, this is the only embedding with guarantees superior to the ones for the accelerated Gaussian binary embedding, in the sense that it is simultaneously a $\del$-binary embedding with the optimal bit complexity and runs in near-linear time in a larger parameter range than $\log N\lesssim \del^{2}n^{1/2}$. Unfortunately, as we point out in Remark~\ref{rem:proofGap}, the proof in \cite{YCP15} contains a subtle gap that cannot be fixed in a trivial way.\par 
The goal of our paper is two-fold. On the one hand, by deriving improved variance bounds we contribute to the theoretical understanding of the numerical experiments with circulant binary embeddings from \cite{YKG14}. On the other hand, our work largely fixes the proof gap in \cite{YCP15}. We start by considering a binary embedding $f_A$ with $A=R_IC_g$. We show, on the one hand, that there are two vectors $p,q \in S^{n-1}$ that cannot be embedded well for \emph{any} choice of $I\subset [n]$. On the other hand, we prove that if $I$ consists of $m$ dyadic integers and $p,q$ are $m$-sparse, then 
$$\operatorname{Var}\big(d_H(f_A(p),f_A(q))\big)\lesssim \frac{1}{m}.$$
Thus, in general a subsampled Gaussian circulant matrix does \emph{not} yield a good binary embedding if one does not randomize the column signs of the matrix. However, by subsampling judiciously, one can still obtain optimal variance decay for relatively sparse vectors. Next, we consider the situation where we additionally use column sign randomization, i.e., we set $A=R_IC_gD_{\eps}$. In Theorem~\ref{varbound} we show that if $I$ is a set of $m$ indices chosen uniformly at random, then for any $p,q\in \Sp^{n-1}$  
\begin{equation}
\label{eqn:varBoundIntro}
\operatorname{Var}\big(d_H(f_A(p),f_A(q))\big) \lesssim \frac{1}{m} + \frac{1}{\sqrt{n}},
\end{equation}
which is optimal for $m\leq \sqrt{n}$. Up to the $1/\sqrt{n}$ factor, this is the variance bound claimed and needed for the proof in \cite{YCP15}. As a consequence, we can resurrect the aforementioned guarantees for the second binary embedding map from \cite{YCP15}, at least when $n\gtrsim \del^{-6}$. Our result (\ref{eqn:varBoundIntro}) is always superior to the variance bounds based on well-spreadness conditions from \cite{YBK15} (i.e., even if $p,q$ are maximally well-spread). We conjecture that the $1/\sqrt{n}$ factor in (\ref{eqn:varBoundIntro}) can be removed. This would remove the condition $n\gtrsim \del^{-6}$.\par 
As a final contribution of our work, we analyze modifications of the two embeddings from \cite{YCP15} that are obtained by replacing the fast Johnson-Lindenstrauss transform $\Phi_{\operatorname{FJL}}$ in these maps by a sparse Johnson-Lindenstrauss transform $\Phi_{\operatorname{SJL}}$ \cite{DKS10,BOR10,KaN14}. The resulting embeddings again realize the optimal bit complexity and have near-linear running times on $\sqrt{n}$-sparse vectors under similar conditions on $n,\del$ and $N$. In contrast to the $\Phi_{\operatorname{FJL}}$-based maps, however, the running times can improve even further if the vectors in $\cD$ are very sparse. We refer to Section~\ref{sec:binaryEmbedding} for further details.    

\subsection{Notation}

We start by fixing some notation that will be used throughout. For any given $x=(x_1,\ldots,x_{n})\in \R^n$, we define $D_x\in \R^{n\ti n}$ to be the diagonal matrix with $x$ on its diagonal and we let $C_x\in \R^{n\ti n}$ be the circulant matrix generated by $x$. That is,
$$D_x = \begin{pmatrix}
x_1 & 0 & \cdots & 0 & 0\\
0 & x_1 & 0 & & 0\\
0 & 0 & x_2 & \ddots & \vdots\\
\vdots & &  \ddots & \ddots & 0\\
0 & \cdots & & 0 & x_{n}
\end{pmatrix},\ \ \ C_x=
\begin{pmatrix}
x_{n} & x_{1} & x_2 & \cdots & x_{n-2} & x_{n-1}\\
x_{n-1} & x_n & x_1 &  \cdots & x_{n-3} & x_{n-2}\\
x_{n-2} & x_{n-1} & x_n & \cdots & x_{n-4} & x_{n-3}\\
\vdots & \vdots & \vdots&  \vdots & \vdots & \vdots\\
x_1 & x_2 & x_3 & \cdots & x_{n-1} & x_n
\end{pmatrix}
$$ 
We use the shorthand notation $[n]=\{1,\ldots,n\}$. For $i,j\in [n]$, we will always take $i+j$ to mean addition modulo $n$. If we define the shift operator $T:\R^n\rightarrow \R^n$ by
$$(Tx)_i = x_{i+1},$$
then we can write the $i$-th row of $C_x$ as $T^{n-i}x$. Closely related to circulant matrices are Toeplitz matrices. For a vector $y\in \R^{2n-1}$ the Toeplitz matrix $T_y\in \R^{n\ti n}$ generated by $y$ is defined by
$$T_y=
\begin{pmatrix}
y_{2n-1} & y_{1} & y_2 & \cdots & y_{n-2} & y_{n-1}\\
y_{2n-2} & y_{2n-1} & y_1 & \cdots & y_{n-3} & y_{n-2}\\
y_{2n-3} & y_{2n-2} & y_{2n-1}& \cdots & \cdots & y_{n-3}\\
\vdots & \vdots & \vdots & \vdots & \vdots & \vdots\\
y_n & y_{n+1} & y_{n+2} & \cdots & y_{2n-2} & y_{2n-1}.
\end{pmatrix}
$$ 
Observe that $T_y$ can be viewed as the upper left $n\ti n$ block of the circulant matrix $C_y \in \R^{2n\ti 2n}$. As a final building block for the construction of binary embeddings we define, for any given $I\subset [n]$, the restriction operator $R_I:\R^n\rightarrow \R^{|I|}$ by 
$$R_I x = (x_i)_{i\in I}.$$
Below we will use several different distances. For $x,y\in \R^n$ with $x,y\neq 0$ we define 
$$d_{\Sp^{n-1}}(x,y)=\frac{1}{\pi}\arccos\Big(\frac{\skp{x}{y}}{\eu{x}\eu{y}}\Big).$$ 
In particular, for two points $x,y$ on the unit sphere, $d_{\Sp^{n-1}}(x,y)$ denotes the normalized geodesic distance between $x$ and $y$. The normalization is chosen so that opposite points on the sphere have distance one. On the discrete cube $\{-1,1\}^m$ we consider two different distances. First, we consider the normalized Hamming distance, i.e.,
$$d_H(x,y) = \frac{1}{m} \sum_{i=1}^m 1_{x_i\neq y_i}.$$
Second, we use a distance considered before in \cite{YCP15}. Let $B$ be a given block size so that $m'=m/B$ is an integer. Accordingly, we define $d_{\operatorname{med},B}$ to be the median of the block-wise Hamming distances, that is, for    $x=(x_1,\ldots,x_B)$ and $y=(y_1,\ldots,y_B)$, with $x_i, y_i \in \{-1,1\}^{m'}$ for $i\in [B]$, we set
$$d_{\operatorname{med},B}(x,y) = \med((d_H(x_i,y_i))_{i=1}^B).$$

\subsection{Acknowledgement}

The authors would like to thank the reviewers for valuable comments, in particular for a suggestion that substantially shortened the proof of Lemma~\ref{ab}. A.\ Stollenwerk acknowledges funding by the European Research Council through ERC Starting Grant StG 258926. S.\ Dirksen and A.\ Stollenwerk acknowledge funding by the Deutsche Forschungsgemeinschaft (DFG) through the project Quantized Compressive Spectrum Sensing (QuaCoSS), which is part of the priority program SPP 1798 Compressed Sensing in Information Processing (COSIP).

\section{Variance bounds}

In this section we consider an $m\times n$ subsampled Gaussian circulant matrix $A=R_I C_g$, with $g$ an $n$-dimensional standard Gaussian, as well as a version with randomized column signs, i.e., $A=R_IC_g D_{\eps}$, where $\eps$ is a Rademacher vector, i.e., a vector of $n$ independent random signs ($\bP(\eps_i=1)=\bP(\eps_i=-1)=1/2$). Denote the $i$-th row of $A$ by $a_i$. By the invariance of the standard Gaussian distribution under reflections, we have in both cases $a_i \overset{d}{\sim} g$. Hence for any $p,q\in \Sp^{n-1}$
\begin{align}
\label{eqn:unbEst}
\E (d_H(f_A(p),f_A(q)))&=\frac{1}{m}\sum_{i\in I} \bP (\sgn(\skp{a_i}{p})\neq \sgn(\skp{a_i}{q}))\\
                                   &=\bP (\sgn(\skp{g}{p})\neq \sgn(\skp{g}{q}))=d_{\Sp^{n-1}}(p,q),  \nonumber
\end{align}
where the final equality is \cite[Lemma 3.2]{GoW95}. This shows that $d_H(f_A(p),f_A(q))$ is an unbiased estimator of the geodesic distance. To investigate whether these random matrices induce a $\del$-binary embedding $f_A$ of finite point sets into $(\{-1,1\}^m,d_H)$, we first consider a more simple problem and try to bound the quantity
$$\Var \big(d_H(f_A(p),f_A(q))\big) = \Var \big(d_H(\sgn(Ap),\sgn(Aq))\big)$$
for arbitrary $p,q\in \Sp^{n-1}$. Define the indicator random variable
$$X_i = 1_{\sgn(\skp{a_i}{p}) \neq \sgn(\skp{a_i}{q})}.$$
We can then write
\begin{align}
\label{eqn:varExpand}
\Var \big(d_H(f_A(p),f_A(q))\big) & = \Var\Big(\frac{1}{m}\sum_{i=1}^m X_i\Big) \nonumber \\
& = \frac{1}{m^2}\sum_{i=1}^m \Var(X_i) + \frac{2}{m^2}\sum_{1\leq i<j\leq m} \Cov{X_i}{X_j}. 
\end{align}
If $A$ is a standard Gaussian matrix, then the second term vanishes and it is easy to obtain an optimal bound. In the case of Gaussian circulant matrices, the rows of $A$ are heavily dependent and it is therefore non-trivial to bound the cross-terms $\Cov{X_i}{X_j}$. This is the main technical challenge addressed in this section.\par
It is worthwhile to note that the best possible variance decay in $m$ that we can expect for arbitrary $p,q\in \Sp^{n-1}$ is 
\begin{equation}
\label{eqn:optVarDecay}
\Var \big(d_H(f_A(p),f_A(q))\big)\lesssim \frac{1}{m}.
\end{equation}
Indeed, by Chebyshev's inequality, this estimate implies
$$\bP(|d_H(f_A(p),f_A(q)) - d_{\Sp^{n-1}}(p,q)|\geq \del)\lesssim \frac{1}{m \del^2}$$ 
and in particular, on every two-point set $\{p,q\}$ the map $f_A$ is an (oblivious) $\del$-binary embedding with probability $1/2$, say, if $m\gtrsim \del^{-2}$. By \cite[Theorem 3.1]{YCP15}, this scaling in $\del$ is optimal.\par
In the first part of this section, we investigate variance bounds for a subsampled Gaussian circulant matrix. It turns out that there exist vectors $p,q$ for which the variance does not decay at all. However, one can still get optimal decay for sparse vectors if one subsamples dyadically (see Proposition~\ref{withoutrad}). In the second part of this section, we consider a subsampled Gaussian circulant matrix with randomized column signs. Our main result, Theorem~\ref{varbound}, says that (\ref{eqn:optVarDecay}) holds for $m\leq \sqrt{n}$ if the set of sampled rows $I$ is chosen uniformly at random. This result is a key ingredient for our results on fast binary embeddings in Section~\ref{sec:binaryEmbedding}.\par 
For any $x\in \R^n$ we define the binary random variable
$$Z(x)=\sgn(\skp{g}{x}).$$
The following is a key observation to control the covariance terms $\Cov{X_i}{X_j}$ for $i\neq j$.
\begin{thm}\label{Cov} For any $x_1, x_2, y_1, y_2\in \Sp^{n-1}$
\begin{equation*}
|\Cov{1_{Z(x_1)\neq Z(x_2)}}{1_{Z(y_1)\neq Z(y_2)}}| \leq 8 \max\{|\skp{x_1}{y_1}|, |\skp{x_1}{y_2}|, |\skp{x_2}{y_1}|, |\skp{x_2}{y_2}|\}.
\end{equation*}
\end{thm}
\begin{rem}
In \cite[Lemma 6]{YBK15}, it was shown that 
\begin{equation}
\label{eqn:YBK6}
|\Cov{1_{Z(x_1)\neq Z(x_2)}}{1_{Z(y_1)\neq Z(y_2)}}|\leq 2\max\{\|\Pi y_1\|_2,\|\Pi y_2\|_2\},
\end{equation}
where $\Pi$ is the projection onto $\Span\{x_1,x_2\}$. If $\{x_1,z_2\}$ is any orthonormal basis of $\Span\{x_1,x_2\}$, then the right hand side is up to constants equal to 
$$\max\{|\skp{x_1}{y_1}|, |\skp{x_1}{y_2}|, |\skp{z_2}{y_1}|, |\skp{z_2}{y_2}|\}.$$
Note that this is, again up to constants, larger than the right hand side in Theorem~\ref{Cov}. Indeed, as $x_2=\langle x_1,x_2\rangle x_1 + \langle z_2,x_2\rangle z_2$, it follows for $i=1,2$,
$$|\langle x_2,y_i\rangle| \leq |\langle x_1,y_i\rangle| + |\langle z_2,y_i\rangle|.$$
In particular, (\ref{eqn:YBK6}) is not a good bound if $x_1,x_2$ are far from orthogonal. For example, let $x_1=e_1$, $x_2=(\cos \al,\sin\al)$, $y_1=e_2$ and $y_2=e_3$. The right hand side in Theorem~\ref{Cov} is equal to $10|\skp{x_2}{y_1}|=10|\sin \al|$, which scales as $\al$ if $\al$ is small. On the other hand, $\|\Pi y_1\|_2=1$, so that (\ref{eqn:YBK6}) gives a trivial bound.
\end{rem}
Theorem~\ref{Cov} can more easily be shown in the case where $\skp{x_1}{x_2}=0$. In order to tackle the non-orthogonal case, we use Lemma~\ref{ab} below. 
\begin{lemma}\label{ab} Let $g_1$, $g_2$ and $g_3$ be independent, standard Gaussian random variables. For $a,b\in \R$ define
\begin{align*}
f(a,b)=\Pb( \sgn(g_1) \neq \sgn(g_1+ ag_3), \sgn(g_2)\neq \sgn(g_2 + bg_3)).
\end{align*}
Then
\begin{align}\label{abineq}
f(a,b)\leq \frac{1}{2\pi}  |ab|.
\end{align}
\end{lemma}
\begin{proof} By invariance of the standard Gaussian distribution under reflections, i.e.,
$$(g_1,g_2,g_3) \overset{d}{\sim} (\eps_1 g_1, \eps_2 g_2, \eps_3 g_3) \text{ for all } \eps_1, \eps_2, \eps_3 \in \{\pm 1\},$$ 
we may assume that $a,b\geq 0$. Moreover, using invariance of the Gaussian distribution under reflection several times, we obtain
\begin{align*}
f(a,b) &= \Pb( g_1\geq 0,  g_1+ ag_3\leq 0, \sgn(g_2)\neq \sgn(g_2 + bg_3))\\
	 &\qquad \qquad \qquad +\Pb( g_1\leq 0,  g_1+ ag_3\geq 0, \sgn(g_2)\neq \sgn(g_2 + bg_3))\\
	 &= \Pb( g_1\geq 0,  g_1+ ag_3\leq 0, \sgn(g_2)\neq \sgn(g_2 + bg_3))\\
	 &\qquad \qquad \qquad +\Pb( -g_1\geq 0,  -g_1- ag_3\leq 0, \sgn(-g_2)\neq \sgn(-g_2 - bg_3))\\
	 &= 2\Pb( g_1\geq 0,  g_1+ ag_3\leq 0, \sgn(g_2)\neq \sgn(g_2 + bg_3))\\
          &= 2\Big(\Pb( g_1\geq 0,  g_1+ ag_3\leq 0, g_2\geq 0, g_2 + bg_3\leq 0)\\ 
          &  \qquad \qquad \qquad + \Pb( g_1\geq 0,  g_1+ ag_3\leq 0, g_2\leq 0, g_2 + bg_3\geq 0) \Big)\\
          &= 2\Big( \Pb( g_1\geq 0,  g_1- ag_3\leq 0, g_2\geq 0, g_2 - bg_3\leq 0)\\ 
          & \qquad \qquad \qquad + \Pb( g_1\geq 0,  g_1- ag_3\leq 0, g_2\leq 0, g_2 - bg_3\geq 0) \Big).
\end{align*}
Clearly, the second term vanishes and so
\begin{align*}
f(a,b)  &=2\Pb( 0\leq g_1\leq a g_3, 0\leq g_2 \leq bg_3).
\end{align*}
Notice that if $g$ is standard Gaussian, then for any $t\geq 0$  
$$\Pb(0\leq g\leq t)\leq \tfrac{1}{\sqrt{2\pi}}t.$$
Using this inequality we find
\begin{align*}
\Pb( 0\leq g_1\leq a g_3, 0\leq g_2 \leq bg_3)&=\E(1_{0\leq g_3}1_{0\leq g_1\leq a g_3}1_{0\leq g_2 \leq bg_3})\\
&=\E_{g_3}(1_{0\leq g_3}\E_{g_1,g_2}(1_{0\leq g_1\leq a g_3}1_{0\leq g_2 \leq bg_3}))\\
&=\E_{g_3}(1_{0\leq g_3}\E_{g_1}(1_{0\leq g_1\leq a g_3})\E_{g_2}(1_{0\leq g_2 \leq bg_3}))\\
&\leq \E_{g_3}(1_{0\leq g_3}(\tfrac{1}{\sqrt{2\pi}}a g_3)(\tfrac{1}{\sqrt{2\pi}}bg_3))\\
&=\tfrac{ab}{2\pi}\E(1_{0\leq g_3}(g_3)^2)=\tfrac{ab}{4\pi},
\end{align*}     
which implies
$f(a,b)\leq \frac{ab}{2\pi}$.  
\end{proof}
\begin{rem} The constant $\tfrac{1}{2\pi}$ in inequality \eqref{abineq} of Lemma~\ref{ab} is optimal. Indeed, consider the symmetric case $a=b$. From the proof of Lemma~\ref{ab} we see that
\begin{align*}
\label{eqn:faaFormula}
f(a,a)&=2\Pb( 0\leq g_1\leq a g_3, 0\leq g_2 \leq ag_3)\\
&=2\Pb(\{0\leq g_1\leq g_2\leq a g_3\}\cup\{0\leq g_2\leq g_1\leq a g_3\} )\\
&=2\big(\Pb(0\leq g_1\leq g_2\leq a g_3)+\Pb(0\leq g_2\leq g_1\leq a g_3)\big)\\
&=4\Pb(0\leq g_1\leq g_2\leq a g_3)\\
&=4\int_{0}^{\infty} \frac{1}{\sqrt{2\pi}} e^{-x_1^2/2}dx_1 \int_{x_1}^{\infty} \frac{1}{\sqrt{2\pi}} e^{-x_2^2/2}dx_2 \int_{\frac{x_2}{a}}^{\infty} \frac{1}{\sqrt{2\pi}} e^{-x_3^2/2}dx_3.
\end{align*}
Observe that
\begin{align*}
\frac{df}{da}(a,a)  &= \frac{4}{(2\pi)^{3/2}}   \int_{0}^{\infty} e^{-x_1^2/2} dx_1 \int_{x_1}^{\infty}  e^{-x_2^2/2}(-e^{- x_2^2/2a^2}) (-\frac{x_2}{a^2})dx_2\\
                            &= \frac{4}{(2\pi)^{3/2} a^2} \int_{0}^{\infty}  e^{-x_1^2/2}dx_1 \int_{x_1}^{\infty}  e^{-\frac{x_2^2}{2}(\frac{a^2+1}{a^2})} x_2 dx_2\\
                            &= \frac{4}{(2\pi)^{3/2} a^2} \int_{0}^{\infty}  e^{-x_1^2/2} \Big(\frac{a^2}{a^2+1}\Big) e^{-\frac{x_1^2}{2}(\frac{a^2+1}{a^2})}dx_1\\
                            &=\frac{4}{(2\pi)^{3/2} (a^2+1)} \int_{0}^{\infty}   e^{-\frac{x_1^2}{2}(\frac{2a^2+1}{a^2})}dx_1\\
                            &=\frac{4}{(2\pi)^{3/2} (a^2+1)} \sqrt{\frac{\pi}{2}} \frac{a}{\sqrt{2a^2+1}} = \frac{1}{\pi}\frac{a}{(a^2+1)\sqrt{2a^2+1}}.
\end{align*}
Hence, 
\begin{equation}
\label{flimit}
\lim_{a\to 0}\frac{f(a,a)}{a^2} =\lim_{a\to 0}\frac{\frac{df}{da}(a,a)}{2a} = \lim_{a\to 0} \frac{1}{2\pi}\frac{1}{(a^2+1)\sqrt{2a^2+1}}=\frac{1}{2\pi}.
\end{equation}
\end{rem}
In the proof of Theorem~\ref{Cov} we will use the following simple estimates. 
\begin{lemma}
\label{geo} 
Suppose $x,y\in \Sp^{n-1}$ with $\skp{x}{y}\geq 0$. Then
\begin{align}\label{geoin1}
d_{\Sp^{n-1}}(x,y)\leq \frac{1}{2^{3/2}} \eu{x-y}
\end{align}
and, as a consequence,
\begin{align}\label{geoin2}
d_{\Sp^{n-1}}(x,y)\leq \frac{1}{2} \sqrt{1-\skp{x}{y}^2}.
\end{align}
\end{lemma}
\begin{proof} \eqref{geoin2} immediately follows from \eqref{geoin1} as
$$\frac{1}{2^{3/2}} \eu{x-y}=\frac{1}{2^{3/2}} \sqrt{2-2\skp{x}{y}}=\frac{1}{2}\sqrt{1-\skp{x}{y}}\leq \frac{1}{2}\sqrt{1-\skp{x}{y}^2}.$$
To prove \eqref{geoin1}, define $z=\skp{x}{y}\in [0,1]$. Then
$d_{\Sp^{n-1}}(x,y)=\frac{1}{\pi}\arccos(z)$ and\\ $\frac{1}{2^{3/2}} \eu{x-y}=\frac{1}{2} \sqrt{1-z}$. We need to show that 
$f(z)=\frac{1}{2} \sqrt{1-z}- \frac{1}{\pi}\arccos(z)\geq 0$ for $z\in [0,1]$. Clearly, $f(0)=f(1)=0$ and
$$f'(z)=-\frac{1}{4}\frac{1}{\sqrt{1-z}}+ \frac{1}{\pi} \frac{1}{\sqrt{1-z^2}}=\frac{-\frac{1}{4}\sqrt{1+z} + \frac{1}{\pi}}{\sqrt{1-z^2}}.$$
Hence $f'(z_0)=0$ for $z_0=\frac{16}{\pi^2}-1\in (0,1)$. 
Since $f'(z)>0$ for $z\in (0,z_0)$ and $f'(z)<0$ for $z\in (z_0,1)$, the result follows.
\end{proof}
We are now ready to prove Theorem~\ref{Cov}.
\begin{proof}[Proof of Theorem~\ref{Cov}] 
We may assume that $\skp{x_1}{x_2}\geq 0$. Indeed, if $\skp{x_1}{x_2}\leq 0$, then $\skp{x_1}{-x_2}\geq 0$ and
by using $d_{\Sp^{n-1}}(x_1,x_2)+d_{\Sp^{n-1}}(x_1,-x_2)=1$ and $\Pb(Z(y_1)\neq Z(y_2))=d_{\Sp^{n-1}}(y_1,y_2)$
 we see
\begin{align*}
\Cov{&1_{Z(x_1)\neq Z(-x_2)}}{1_{Z(y_1)\neq Z(y_2)}}\\
 &= \Pb\big( Z(x_1) \neq Z(-x_2),  Z(y_1)\neq Z(y_2)  \big)  -  d_{\Sp^{n-1}}(x_1,-x_2)d_{\Sp^{n-1}}(y_1,y_2)\\
										  &=d_{\Sp^{n-1}}(y_1,y_2)- \Pb\big( Z(x_1) \neq Z(x_2),  Z(y_1)\neq Z(y_2)  \big)\\ 
										  &\quad - \big(1- d_{\Sp^{n-1}}(x_1,x_2)\big)d_{\Sp^{n-1}}(y_1,y_2)\\
										  &=-\Pb\big( Z(x_1) \neq Z(x_2),  Z(y_1)\neq Z(y_2)  \big)+ d_{\Sp^{n-1}}(x_1,x_2)d_{\Sp^{n-1}}(y_1,y_2)\\
										  &=-\Cov{1_{Z(x_1)\neq Z(x_2)}}{1_{Z(y_1)\neq Z(y_2)}}, 
\end{align*}
which implies $$|\Cov{1_{Z(x_1)\neq Z(x_2)}}{1_{Z(y_1)\neq Z(y_2)}}|=|\Cov{1_{Z(x_1)\neq Z(-x_2)}}{1_{Z(y_1)\neq Z(y_2)}}|.$$
From now on, we assume $\skp{x_1}{x_2}\geq 0$. Clearly, $x_1,x_2, y_1, y_2$ span an at most $4$-dimensional subspace of $\R^n$. Hence one can construct an orthogonal matrix $O$, which maps $x_1, x_2, y_1$ and $y_2$ into $\R^4\times \{0\}^{n-4}$. In fact, we can choose $O$ so that
\begin{align*}
Ox_1 &=e_1, \; Ox_2 = a_1e_1+  a_2 e_2=:a, \; Oy_1=b_1e_1 + b_2 e_2 +b_3 e_3=:b,\\
 Oy_2&= c_1e_1+ c_2e_2+c_3e_3+c_4e_4=:c
\end{align*}
with $a_2, b_3, c_4 \geq 0$. This can be achieved by defining $O=O_4 O_3 O_2 O_1$, where 
$$O_1=\tilde{O_1},\quad  O_2 =\begin{pmatrix}
1 & 0 \\ 
0 & \tilde{O_2}  
\end{pmatrix}, \quad
O_3 =\begin{pmatrix}
1 & 0  & 0\\ 
0 & 1 & 0 \\
0 & 0 & \tilde{O_3}  
\end{pmatrix},\quad
O_4 =\begin{pmatrix}
1 & 0  & 0 & 0\\ 
0 & 1 & 0  & 0 \\
0 & 0 & 1 &  0\\
0 & 0 & 0 & \tilde{O_4}  
\end{pmatrix}$$
with $\tilde{O_i}\in \R^{(n+1- i) \times (n + 1 - i)}$ orthogonal matrices, such that 
$$O_1x_1=e_1,\; O_2 O_1x_2=a,\; O_3 O_2 O_1 y_1=b\; \text{ and }O_4 O_3 O_2 O_1 y_2=c.$$
Observe that $0\leq \skp{x_1}{x_2}=\skp{e_1}{a}=a_1$ and  $a_2=\sqrt{1-a_1^2}=\sqrt{1-\skp{x_1}{x_2}^2}$. We can assume $a_2>0$. Otherwise the statement is trivial, because then $\skp{x_1}{x_2}=1$, which implies $x_1=x_2$ and hence 
$\Cov{1_{Z(x_1)\neq Z(x_2)}}{1_{Z(y_1)\neq Z(y_2)}}=0$. Moreover, we may assume that $\gamma:=\sqrt{c_3^2+c_4^2}>0$, otherwise $1=\sqrt{c_1^2+c_2^2}\leq |c_1|+|c_2|$ and using inequality \eqref{geoin2} in Lemma~\ref{geo} yields
\begin{align*}
|\Cov{&1_{Z(x_1)\neq Z(x_2)}}{1_{Z(y_1)\neq Z(y_2)}}|\\
        &\leq \Pb\big(Z(x_1)\neq Z(x_2) \big)=d_{\Sp^{n-1}}(x_1,x_2)\leq \frac{\sqrt{1-\skp{x_1}{x_2}^2}}{2}\\
         &= \frac{a_2}{2} \leq \frac{a_2(|c_1|+|c_2|)}{2}\leq \frac{|c_1|+|a_2 c_2|}{2}= \frac{|\skp{e_1}{c}| + |\skp{a}{c} - a_1c_1|}{2}\\
         &\leq \frac{2|\skp{e_1}{c}| + |\skp{a}{c}|}{2}=\frac{2|\skp{x_1}{y_2}| + |\skp{x_2}{y_2}|}{2}.
\end{align*}
The idea is now to replace $b,c$ by suitable approximants $b',c'$, which are orthogonal to $e_1,a$. Set $b'=e_3$ and $c'=\frac{1}{\gamma} (c_3e_3+ c_4e_4)$. Clearly, $b',c'\in \Sp^{n-1}$ are orthogonal to $e_1, a$ and we can control their distance to the original vectors $b$ and $c$. Indeed, since $\skp{b}{b'}=b_3\geq 0$ and $\skp{c}{c'}=\gamma\geq 0$, inequality \eqref{geoin2} in Lemma~\ref{geo} yields 
\begin{align*} 
d_{\Sp^{n-1}}(b,b')\leq \frac{1}{2} \sqrt{1-\skp{b}{b'}^2} =\frac{1}{2} \sqrt{1-b_3^2}= \frac{1}{2} \sqrt{b_1^2+b_2^2}\leq \frac{1}{2} ( |b_1|+|b_2|)
\end{align*}
and
\begin{align*} 
d_{\Sp^{n-1}}(c,c')\leq \frac{1}{2}  \sqrt{1-\skp{c}{c'}^2} = \frac{1}{2} \sqrt{1-(c_3^2+c_4^2)}= \frac{1}{2}  \sqrt{c_1^2+c_2^2}\leq  \frac{1}{2}  (|c_1|+|c_2|).
\end{align*}
Define $$\eps = \max\{|\skp{x_1}{y_1}|, |\skp{x_1}{y_2}|, |\skp{x_2}{y_1}|, |\skp{x_2}{y_2}|\}.$$
Then
$$|b_1|=|\skp{e_1}{b}|=|\skp{x_1}{y_1}|\leq \eps \text{ as well as } |c_1|=|\skp{e_1}{c}|=|\skp{x_1}{y_2}|\leq \eps.$$
Similarly we can estimate
\begin{equation}
\label{eqn:a2b2Est}
|a_2b_2|=|\skp{a}{b}- a_1b_1|\leq |\skp{a}{b}|+ |a_1b_1|= |\skp{x_2}{y_1}|+ |a_1b_1|\leq 2\eps
\end{equation}
and
\begin{equation}
\label{eqn:a2c2Est}
|a_2c_2|=|\skp{a}{c}- a_1c_1|\leq |\skp{a}{c}|+ |a_1c_1|= |\skp{x_2}{y_2}|+ |a_1c_1|\leq 2\eps,
\end{equation}
which yields  $$|b_2|\leq \frac{2\eps}{a_2} \text{ as well as }  |c_2|\leq \frac{2\eps}{a_2}.$$
In summary we obtain the bounds
\begin{align}\label{dbdc}
d_{\Sp^{n-1}}(b,b')\leq \frac{\eps}{2}+ \frac{\eps }{a_2} \text{ and }  d_{\Sp^{n-1}}(c,c')\leq \frac{\eps}{2}+ \frac{\eps }{a_2}.
\end{align}
The rotational invariance of the standard Gaussian distribution implies
\begin{align*}
\Cov{1_{Z(x_1)\neq Z(x_2)}}{1_{Z(y_1)\neq Z(y_2)}}=\Cov{1_{Z(e_1)\neq Z(a)}}{1_{Z(b)\neq Z(c)}}.																	      
\end{align*}
Moreover, since $\skp{e_1}{b'}= \skp{e_1}{c'}= \skp{a}{b'}=\skp{a}{c'}=0$, $1_{Z(e_1)\neq Z(a)}$ and $1_{Z(b')\neq Z(c')}$ are independent and in particular
$$\Cov{1_{Z(e_1)\neq Z(a)}}{1_{Z(b')\neq Z(c')}} = 0.$$
Using the triangle inequality we now obtain
\begin{align}\label{threesummands}
 |\Cov{&1_{Z(x_1)\neq Z(x_2)}}{1_{Z(y_1)\neq Z(y_2)}}|\\
 &= |\Cov{1_{Z(e_1)\neq Z(a)}}{1_{Z(b)\neq Z(c)}} - \Cov{1_{Z(e_1)\neq Z(a)}}{1_{Z(b')\neq Z(c')}}| \nonumber \\
 &\leq |\Pb\big( Z(e_1) \neq Z(a),  Z(b)\neq Z(c)  \big)- \Pb\big( Z(e_1) \neq Z(a),  Z(b')\neq Z(c')  \big) | \nonumber  \\
&\quad+ |d_{\Sp^{n-1}}(e_1,a)d_{\Sp^{n-1}}(b,c)- d_{\Sp^{n-1}}(e_1,a)d_{\Sp^{n-1}}(b',c')| \nonumber  \\
&\leq \Pb\big(Z(e_1) \neq Z(a), Z(b)\neq Z(b')\big) + \Pb\big(Z(e_1) \neq Z(a), Z(c)\neq Z(c')\big) \nonumber\\
&\quad+ d_{\Sp^{n-1}}(e_1,a)\big(d_{\Sp^{n-1}}(b,b')+ d_{\Sp^{n-1}}(c,c')\big). \nonumber
\end{align}   
In the last step, we have used 
$$d_{\Sp^{n-1}}(b,c)\leq d_{\Sp^{n-1}}(b,b') + d_{\Sp^{n-1}}(b',c')+ d_{\Sp^{n-1}}(c',c),$$ 
and that by decomposing the event
\begin{align*}
\{Z(b)\neq Z(c)\} &= \{Z(b)\neq Z(c), Z(b)=Z(b'), Z(c)=Z(c')\}\\
                           &\cup \{Z(b)\neq Z(c), \big(Z(b)\neq Z(b') \vee Z(c)\neq Z(c')\big)\},
\end{align*} 
we obtain
\begin{align*}
&\;\;\;\;\; \Pb\big( Z(e_1) \neq Z(a),  Z(b)\neq Z(c) \big) \leq  \Pb\big( Z(e_1) \neq Z(a),  Z(b')\neq Z(c')\big)\\
&\quad + \Pb\big(Z(e_1) \neq Z(a), Z(b)\neq Z(b')\big) + \Pb\big(Z(e_1) \neq Z(a), Z(c)\neq Z(c')\big).
\end{align*}
It remains to bound all three summands appearing on the far right hand side of \eqref{threesummands}. We estimate using \eqref{dbdc} and Lemma \ref{geo}
\begin{align*}
d_{\Sp^{n-1}}(e_1,a)\big(d_{\Sp^{n-1}}(b,b')+ d_{\Sp^{n-1}}(c,c')\big) &\leq d_{\Sp^{n-1}}(e_1,a)\Big(\eps + \frac{2\eps}{a_2}\Big)\\
                                                                                                            &\leq d_{\Sp^{n-1}}(e_1,a)\eps + \eps\leq 2\eps.
\end{align*}
Before estimating the other two summands in \eqref{threesummands}, observe that with $$k_b=\sqrt{b_2^2+b_3^2} \text{ and } k_c=\sqrt{c_2^2+c_3^2+c_4^2},$$
we obtain using Lemma~\ref{geo},
\begin{align}\label{rb}
d_{\Sp^{n-1}}\Big(b,\Big(0,\frac{b_2}{k_b},\frac{b_3}{k_b}\Big)\Big)\leq \frac{1}{2} \sqrt{1- \Big(\frac{b_2^2}{k_b}+ \frac{b_3^2}{k_b}\Big)^2}=\frac{1}{2} |b_1|\leq \frac{\eps}{2}
\end{align}
and similarly, 
\begin{align}\label{rc}
d_{\Sp^{n-1}}\Big(c, \Big(0,\frac{c_2}{k_c},\frac{c_3}{k_c},\frac{c_4}{k_c} \Big)\Big)\leq \frac{\eps}{2}.
\end{align}
Notice that if $k_b=0$ (respectively, $k_c=0$), then $b=\pm e_1$ (respectively, $c=\pm e_1$) and hence $\eps=1$, so that the result is trivial in this case.\par
Let us now estimate $\Pb\big(Z(e_1) \neq Z(a), Z(b)\neq Z(b')\big)$, the second summand on the far right hand side of \eqref{threesummands} can be bounded in the same fashion. We distinguish three cases.
First, if $a_2^2>\frac{1}{4}$, then using \eqref{dbdc}
$$\Pb\big(Z(e_1) \neq Z(a), Z(b)\neq Z(b')\big)\leq d_{\Sp^{n-1}}(b,b')\leq \frac{\eps}{2} + \frac{\eps}{a_2} \leq 3\eps.$$
Second, if $b_1^2+b_2^2> \frac{3}{4}$, then $$d_{\Sp^{n-1}}(b,b')=\frac{1}{\pi}\arccos(b_3)=\frac{1}{\pi}\arccos\Big(\sqrt{1-(b_1^2+b_2^2)}\Big)\geq \frac{1}{\pi}\arccos(1/2)\geq \frac{1}{\pi},$$ which implies by \eqref{dbdc}
\begin{align*}
\Pb\big(Z(e_1) \neq Z(a), Z(b)\neq Z(b')\big) \leq d_{\Sp^{n-1}}(e_1,a) &\leq  \frac{1}{2} \sqrt{1-\skp{e_1}{a}^2}\\
                                                                                                               &=\frac{a_2}{2} \leq \frac{3\eps}{4d_{\Sp^{n-1}}(b,b')}\leq \frac{3\pi \eps}{4}\leq 3\eps.
\end{align*}
Finally, suppose that $a_2^2\leq \frac{1}{4}$ and $b_1^2+b_2^2\leq \frac{3}{4}$. This implies
\begin{equation}
\label{eqn:fraca2b2}
\frac{|a_2b_2|}{a_1 b_3}=\frac{|a_2b_2|}{\sqrt{1-a_2^2} \sqrt{1-b_1^2-b_2^2}}\leq \frac{4}{\sqrt{3}}|a_2b_2| \text{ and } a_1, b_3>0.
\end{equation}
Using \eqref{rb} we find
\begin{align*}
& \Pb\big(Z(e_1) \neq Z(a), Z(b)\neq Z(b')\big) \\
&\qquad = \Pb\big(\sgn(g_1)\neq \sgn(a_1g_1+a_2g_2), \sgn(b_1g_1+b_2g_2+b_3g_3)\neq \sgn(g_3)\big)\\
                        							      &\qquad \leq  \Pb\big(\sgn(g_1)\neq \sgn(a_1g_1+a_2g_2), \sgn(b_2g_2+b_3g_3)\neq \sgn(g_3)\big)\\
							                     &\qquad \quad+ \Pb\big(\sgn(b_1g_1+b_2g_2+b_3g_3))\neq \sgn(b_2g_2+b_3g_3)\big)\\
							                     &\qquad \leq \Pb\Big(\sgn(g_1)\neq \sgn\Big(g_1+ \frac{a_2}{a_1}g_2\Big), \sgn(g_3)\neq \sgn\Big(\frac{b_2}{b_3}g_2+g_3\Big)\Big)\\
							                     &\qquad \quad + d_{\Sp^{n-1}}\Big(b,\Big(0,\frac{b_2}{k_b},\frac{b_3}{k_b}\Big)\Big)\\
							                     &\qquad \leq \Pb\Big(\sgn(g_1)\neq \sgn\Big(g_1+ \frac{a_2}{a_1}g_2\Big), \sgn(g_3)\neq \sgn\Big(\frac{b_2}{b_3}g_2+g_3\Big)\Big) +\frac{\eps}{2}.\\
\end{align*}
Applying Lemma~\ref{ab}, (\ref{eqn:fraca2b2}), and (\ref{eqn:a2b2Est}) we obtain
\begin{align*}
&\Pb\Big(\sgn(g_1)\neq \sgn\Big(g_1+ \frac{a_2}{a_1}g_2\Big), \sgn(g_3)\neq \sgn\Big(g_3+\frac{b_2}{b_3}g_2\Big)\Big) \\
& \qquad \qquad \leq \frac{|a_2b_2|}{2\pi a_1b_3}\leq \frac{2|a_2b_2|}{\pi \sqrt{3}}\leq \frac{4\eps}{\pi \sqrt{3}}.
\end{align*}
As a consequence,
$$\Pb\big(Z(e_1) \neq Z(a), Z(b)\neq Z(b')\big)\leq \frac{4\eps}{\pi \sqrt{3}}+ \frac{\eps}{2}\leq 3\eps.$$
In all three cases we find $$\Pb\big(Z(e_1) \neq Z(a), Z(b)\neq Z(b')\big)\leq 3\eps.$$
To finish the proof, we bound the second summand on the far right hand side of \eqref{threesummands}. If either $a_2^2>\frac{1}{4}$ or $c_1^2+c_2^2> \frac{3}{4}$, then 
$$ \Pb\big(Z(e_1) \neq Z(a), Z(c)\neq Z(c')\big)\leq 3 \eps$$ follows exactly as above, simply by replacing $b$ and $b'$ everywhere by $c$ and $c'$, respectively.
Let us now assume that $a_2^2\leq \frac{1}{4}$ as well as $1-\gamma^2=c_1^2+c_2^2\leq \frac{3}{4}$. Then
\begin{equation}
\label{eqn:fraca2c2}
\frac{|a_2c_2|}{a_1  \gamma}=\frac{|a_2c_2|}{\sqrt{1-a_2^2}\sqrt{1-c_1^2-c_2^2}}\leq \frac{4}{\sqrt{3}}|a_2c_2| \text{ and } a_1, \gamma>0.
\end{equation}
Using \eqref{rc} we estimate
\begin{align*}
&\Pb\big(Z(e_1) \neq Z(a), Z(c)\neq Z(c')\big) \\
& \qquad = \Pb\big(\sgn(g_1)\neq \sgn(a_1g_1+a_2g_2),\\
& \qquad \qquad \qquad \quad \; \sgn(c_1g_1+c_2g_2+c_3g_3 + c_4g_4)\neq \sgn(c_3g_3 + c_4g_4)  \big)\\
&\qquad \leq \Pb\big(\sgn(g_1)\neq \sgn(a_1g_1+a_2g_2), \sgn(c_2g_2+c_3g_3 + c_4g_4)\neq \sgn(c_3g_3 + c_4g_4) \big)\\
&\qquad \quad + \Pb\big(\sgn(c_1g_1+c_2g_2+c_3g_3+c_4g_4)\neq \sgn(c_2g_2+c_3g_3+c_4g_4)\big)\\
&\qquad = \Pb\big(\sgn(g_1)\neq \sgn(a_1g_1+a_2g_2), \sgn(c_2g_2+c_3g_3 + c_4g_4)\neq \sgn(c_3g_3 + c_4g_4) \big)\\
&\qquad \quad + d_{\Sp^{n-1}}\Big(c, \Big(0,\frac{c_2}{k_c},\frac{c_3}{k_c},\frac{c_4}{k_c} \Big)\Big)\\
&\qquad \leq \Pb\Big(\sgn(g_1)\neq \sgn\Big(g_1+ \frac{a_2}{a_1}g_2\Big), \\
&\qquad \qquad \qquad \quad  \sgn\Big(\frac{c_2}{\gamma}g_2+\frac{c_3g_3 + c_4g_4}{\gamma}\Big)\neq \sgn\Big(\frac{c_3g_3 + c_4g_4}{\gamma}\Big) \Big) + \frac{\eps}{2}\\
&\qquad =\Pb\Big(\sgn(g_1)\neq \sgn\Big(g_1+ \frac{a_2}{a_1}g_2\Big),  \sgn\Big(\frac{c_2}{\gamma}g_2+g'\Big)\neq \sgn(g')\Big) + \frac{\eps}{2},
\end{align*}
where in the last step we write $g':=\frac{1}{\gamma}(c_3g_3 + c_4g_4)$. Observe that $g'$ is a standard Gaussian random variable which is independent of $g_1$ and $g_2$. Therefore Lemma~\ref{ab}, (\ref{eqn:fraca2c2}), and (\ref{eqn:a2c2Est}) yield
\begin{align*}
&\Pb\Big(\sgn(g_1)\neq \sgn\Big(g_1+ \frac{a_2}{a_1}g_2\Big),\sgn(g')\neq   \sgn\Big(g'+\frac{c_2}{\gamma}g_2\Big)\Big)\\
&\qquad \qquad \qquad \leq \frac{|a_2c_2|}{2\pi a_1\gamma}\leq \frac{2|a_2c_2|}{\pi \sqrt{3}}\leq \frac{4\eps}{\pi \sqrt{3}},
\end{align*}
which implies that in all cases $ \Pb\big(Z(e_1) \neq Z(a), Z(c)\neq Z(c')\big)\leq 3 \eps$.
\end{proof}
\begin{rem} The dependence on $\eps= \max\{|\skp{x_1}{y_1}|, |\skp{x_1}{y_2}|, |\skp{x_2}{y_1}|, |\skp{x_2}{y_2}|\}$ in Theorem~\ref{Cov} is optimal up to a constant. We argue by contradiction.
Assume $\Cov{X}{Y}:=\Cov{1_{Z(x_1)\neq Z(x_2)}}{1_{Z(y_1)\neq Z(y_2)}}\leq C \eps^{s}$ for some constant $C>0$ and $s>1$. For $a>0$ define
 $x_1=e_1$, $x_2=\frac{1}{\sqrt{1+a^2}}(e_1+ a e_3)$, $y_1=e_2$ and $y_2=\frac{1}{\sqrt{1+a^2}}(e_2+ a e_3)$.
Then $\eps=\skp{x_2}{y_2}=\frac{a^2}{1+a^2}$ and
\begin{align*}
\Cov{X}{Y}&= \Pb \big( \sgn(g_1)\neq \sgn(g_1+ ag_3), \sgn(g_2) \neq \sgn(g_2+ ag_3) \big)- \E X \E Y\\
                &= f(a,a)-  g(a)^2,
 \end{align*}
where $f(a,a)$ as in Lemma~\ref{ab} and 
$$g(a)=\E X=\E Y=d_{\Sp^{n-1}}(x_1, x_2)= d_{\Sp^{n-1}}(y_1,y_2)= \frac{1}{\pi}  \arccos \Big(\frac{1}{\sqrt{1+a^2}}\Big).$$
We know $\lim_{a \to 0} \frac{f(a,a)}{a^2}=\frac{1}{2\pi}$ by \eqref{flimit}. Clearly, 
$$\Big(\arccos (\frac{1}{\sqrt{1+a^2}}) \Big)'=\frac{1}{1+a^2}$$
and as a consequence,
\begin{align*}
\lim_{a \to 0}  \frac{\pi^2 g(a)^2}{a^2} &=\lim_{a \to 0} \frac{\arccos^2 (\frac{1}{\sqrt{1+a^2}})}{a^2}\\
							& =\lim_{a \to 0} \frac{\arccos (\frac{1}{\sqrt{1+a^2}}) }{a(1+a^2)} =\lim_{a \to 0} \frac{1}{(1+3a^2)(1+a^2)}=1.
\end{align*}
A contradiction now follows by observing that
$$\lim_{a \to 0} \frac{\Cov{X}{Y}}{a^2}= \frac{1}{2\pi}- \frac{1}{\pi^2}>0,$$
but 
$$\lim_{a\to 0}\frac{C \eps^s}{a^2}=\lim_{a\to 0}\frac{C a^{2s}}{(1+a^2)^s a^2}=\lim_{a\to 0}\frac{C a^{2s-2}}{(1+a^2)^s}=0.$$
\end{rem}
We can now state our main variance bounds for subsampled Gaussian circulant matrices.
\begin{prop}\label{withoutrad} Let $I\subset [n]$, let $g$ be an $n$-dimensional standard Gaussian and set $A=R_I C_g$. 
\begin{enumerate}
\item Set $n=2k$, $p=\sum_{i=1}^k e_{2i-1}$, $q=\sum_{i=1}^k e_{2i}$ and $I\subset [n]$ arbitrary subset with $m$ elements. Then
         $$\Var\big(d_H(f_A(p), f_A(q))\big)=\frac{1}{4}.$$
         In particular, $\bP(|d_H(f_A(p),f_A(q))-d_{\Sp^{n-1}}(p,q)|\geq 1/4)\geq \frac{1}{36}$.
\item Let $I\subset [n]$ be the set of the first $m$ dyadic integers. Then for any two s-sparse vectors $p,q\in \Sp^{n-1}$   
         $$\Var \big(d_H(f_A(p),f_A(q))\big)\lesssim \frac{1}{m}+ \frac{s}{m^2}.$$
\end{enumerate}
\end{prop}
This shows that a subsampled Gaussian circulant matrix in general does not induce a binary embedding with high probability for arbitrary $p,q\in \Sp^{n-1}$. However, if $p,q$ are $m$-sparse and one subsamples dyadically, then $\Var \big(d_H(f_A(p), f_A(q))\big)$ decays optimally in terms of the number of measurements $m$.
\begin{proof}[Proof of Proposition~\ref{withoutrad}]
First, let us consider case $(1)$. As before, we let $a_i$ be the $i$-th row of $A$ and use the notation
$$X_i=1_{\sgn(\skp{a_i}{p}) \neq \sgn(\skp{a_i}{q})}.$$ 
For any $i,j\in I$
\begin{align*}
\Cov{X_i}{X_j} &=\Pb \big( \sgn(\skp{T^{n-i}g}{p})\neq \sgn(\skp{T^{n-i}g}{q}),\\
                       &\qquad \qquad \sgn(\skp{T^{n-j}g}{p}) \neq \sgn(\skp{T^{n-j}g}{q}) \big) - d_{\Sp^{n-1}}(p,q)^2\\
                       &=\Pb \Big( \sgn\Big(\sum_{s=1}^k g_{2k+2s- i -1}\Big)\neq \sgn\Big(\sum_{s=1}^k g_{2k+2s-i}\Big),\\
                       &\qquad \qquad \sgn\Big(\sum_{s=1}^k g_{2k+2s-j-1}\Big) \neq \sgn\Big(\sum_{s=1}^k g_{2k+2s-j}\Big) \Big) - d_{\Sp^{n-1}}(p,q)^2\\
                       &=\Pb \Big( \sgn\Big(\sum_{s=1}^k g_{2k+2s- i -1}\Big)\neq \sgn\Big(\sum_{s=1}^k g_{2k+2s-i}\Big)\Big) - d_{\Sp^{n-1}}(p,q)^2\\
                       &=d_{\Sp^{n-1}}(p,q)-d_{\Sp^{n-1}}(p,q)^2.
                       \end{align*}
Here we used that for any pair $i,j\in [n]$
\begin{align*}
&\Big\{\sgn\Big(\sum_{s=1}^k g_{2k+2s- i -1}\Big)\neq \sgn\Big(\sum_{s=1}^k g_{2k+2s-i}\Big)\Big\}\\
&\qquad =\Big\{ \sgn\Big(\sum_{s=1}^k g_{2k+2s-j-1}\Big) \neq \sgn\Big(\sum_{s=1}^k g_{2k+2s-j}\Big)\Big\}.
\end{align*}                   
Since $\skp{p}{q}=0$, we conclude
\begin{align*}
\Var \big(d_H(f_A(p),f_A(q))\big) &= \Var \Big( \frac{1}{m} \sum_{i\in I} X_i \Big)=\frac{1}{m^2} \sum_{i,j\in I}  \Cov{X_i}{X_j}\\
						         &=d_{\Sp^{n-1}}(p,q)-d_{\Sp^{n-1}}(p,q)^2= \frac{1}{2}-\frac{1}{4}=\frac{1}{4}.
\end{align*}
By the Paley-Zygmund inequality, this implies for any $\del\leq 1/4$
\begin{align*}
& \bP(|d_H(f_A(p),f_A(q))-d_{\Sp^{n-1}}(p,q)|\geq \del) \\
& \qquad \qquad \qquad \geq \frac{(\Var(d_H(f_A(p),f_A(q)))-\del^2)^2}{\E(d_H(f_A(p),f_A(q))-d_{\Sp^{n-1}}(p,q))^4} \geq \frac{1}{36}.
\end{align*}
Let us consider case $(2)$. Here the $i$-th row of $A$ takes the form $a_i=T^{n-2^{i-1}}g$ and
$$X_i=1_{\sgn(\skp{a_i}{p}) \neq \sgn(\skp{a_i}{q})}=1_{\sgn(\skp{T^{n-2^{i-1}}g}{p}) \neq \sgn(\skp{T^{n-2^{i-1}}g}{q})}.$$
Moreover, since the standard Gaussian distribution is invariant under permutation of coordinates
\begin{align*}
\Cov{X_i}{X_j} &=\Cov{1_{\sgn(\skp{g}{p}) \neq \sgn(\skp{g}{q})}}{1_{\sgn(\skp{g}{T^{2^{j-1}-2^{i-1}}p}) \neq \sgn(\skp{g}{T^{2^{j-1}-2^{i-1}}q})}}\\
                                                            &=\Cov{1_{Z(p) \neq Z(q)}}{1_{Z(T^{2^{j-1}-2^{i-1}}p) \neq Z(T^{2^{j-1}-2^{i-1}}q)}}.
\end{align*}
As a consequence, Theorem~\ref{Cov} implies
\begin{align*}
|\Cov{X_i}{X_j}| &\lesssim |\skp{p}{T^{2^{j-1}-2^{i-1}}p}|+  |\skp{p}{T^{2^{j-1}-2^{i-1}}q}|\\
                         &\qquad + |\skp{q}{T^{2^{j-1}-2^{i-1}}p}|+ |\skp{q}{T^{2^{j-1}-2^{i-1}}q}| \nonumber .\\
\end{align*}
Observe that
\begin{align*}
\sum_{1\leq i<j \leq m} |\skp{p}{T^{2^{j-1}-2^{i-1}}q}|\leq \sum_\ell |p_\ell| \sum_{1\leq i<j \leq m} |q_{\ell+ 2^{j-1}-2^{i-1}}|\leq \sum_\ell |p_\ell| \sqrt{s}\leq s,
\end{align*}
where in the last two inequalities we used Cauchy-Schwarz, $\eu{p}=\eu{q}=1$, the sparsity of $p,q$ and that $\{2^{j-1}-2^{i-1} \; : \; 1\leq i<j \leq m \} $ is injectively contained in $[n]$. Putting our estimates together we find
\begin{align*}
\Var \big(d_H(f_A(p),f_A(q))\big) &\leq \frac{1}{m^2} \Big( \sum_{i=1}^m \Var(X_i) + 2\sum_{1\leq i<j \leq m} |\Cov{X_i}{X_j}| \Big)\\
                                                           &\lesssim \frac{1}{m} + \frac{s}{m^2}.
\end{align*}
\end{proof}
Next, we consider a subsampled Gaussian circulant matrix with randomized column signs $A=R_IC_gD_{\eps}$. In this case,
$$X_i=1_{\sgn(\skp{a_i}{p}) \neq \sgn(\skp{a_i}{q})}=1_{\sgn(\skp{(T^{n-i}g) \od \eps}{p}) \neq \sgn(\skp{(T^{n-i}g) \od \eps}{q})},$$
where $\od$ denotes pointwise multiplication of vectors. 
\begin{lemma}\label{radcov} Let $p, q \in \Sp^{n-1}$. For any $i\neq j$, $|j-i|\neq \frac{n}{2}$
$$| \Cov{X_i}{X_j}|\leq 8 \big(  \eu{p\od T^{j-i}p}+ \eu{p\od T^{j-i}q}+ \eu{q\od T^{j-i}p}+ \eu{q\od T^{j-i}q} \big).$$
\end{lemma}
\begin{proof} We may assume $i<j$. Using the invariance of the standard Gaussian distribution under coordinate permutations
\begin{align*}
 &\Cov{X_i}{X_j}\\
  &= \Cov{1_{\sgn(\skp{T^{n-i}g}{\eps\od p}) \neq \sgn(\skp{T^{n-i}g}{\eps\od q})}}{1_{\sgn(\skp{T^{n-j}g}{\eps\od p}) \neq \sgn(\skp{T^{n-j}g}{\eps\od q})}}\\
 	                &= \Cov{1_{Z(\eps\od p) \neq Z(\eps\od q)}}{1_{Z(T^{j-i}(\eps\od p)) \neq Z(T^{j-i}(\eps\od q))}}.
\end{align*}
Since $g$ and $\eps$ are independent
 \begin{align*}
| \Cov{X_i}{X_j}| &=|\E_{\eps}\Cov{1_{Z(\eps\od p) \neq Z(\eps\od q)}}{1_{Z(T^{j-i}(\eps\od p)) \neq Z(T^{j-i}(\eps\od q))}}|\\
     		         &\leq \E_{\eps} |\Cov{1_{Z(\eps\od p) \neq Z(\eps\od q)}}{1_{Z(T^{j-i}(\eps\od p)) \neq Z(T^{j-i}(\eps\od q))}}|.
\end{align*}    		         
Applying Theorem~\ref{Cov} and Jensen's inequality we find		
\begin{align*}		         
| \Cov{X_i}{X_j}|&\leq 8\; \E_{\eps} \big(  |\skp{\eps\od p}{T^{j-i}(\eps\od p)}|+  |\skp{\eps\od p}{T^{j-i}(\eps\od q)}|\\
		         &\quad + |\skp{\eps\od q}{T^{j-i}(\eps\od p)}|+ |\skp{\eps\od q}{T^{j-i}(\eps\od q)}|\big)\\
                          &\leq 8 \Big( \sqrt{ \E \big( |\skp{\eps\od p}{T^{j-i}(\eps\od p)}|^2 \big) } +  \sqrt{ \E\big( |\skp{\eps\od p}{T^{j-i}(\eps\od q)}|^2 \big) }\\  
                          &\quad + \sqrt{ \E\big(|\skp{\eps\od q}{T^{j-i}(\eps\od p)}|^2 \big) } + \sqrt{ \E \big(  |\skp{\eps\od q}{T^{j-i}(\eps\od q)}|^2 \big) } \Big).
\end{align*}
For $s,t\in [n]$, $s\neq t$ and $k\in [n-1]$ 
$$\E (\eps_s \eps_{s+k} \eps_t \eps_{t+k})=
\begin{cases}
1, \quad \text{ if } s=t\; \text{ or } \;s=t+k, \ t=s+k \\
0, \quad \text{ else}.
\end{cases}$$
Observe that $s=t+k, \ t=s+k$ implies $k=\frac{n}{2}$. Set $k=j-i\in [n-1]$, then $k\neq \frac{n}{2}$ and therefore
\begin{align*}
\E |\skp{\eps\od p}{T^{k}(\eps\od q)}|^2 &= \sum_{s,t=1}^n p_s q_{s+k} p_{t} q_{t+k} \E (\eps_s \eps_{s+k} \eps_t \eps_{t+k} )\\
                                                              &= \sum_{s=1}^n p_s^2 q_{s+k}^2 =\eu{p\od T^kq}^2.
\end{align*}
\end{proof}
\begin{thm}\label{varbound} Let $g$ be an $n$-dimensional standard Gaussian vector and let $\eps$ be an independent vector of independent random signs. Let $A=R_{I} C_g D_{\eps}$.
\begin{enumerate}
\item If $I=[m]$, then for any $p,q\in \Sp^{n-1}$ 
\begin{equation}
         \label{eqn:varboundMainDet}
         \Var \big(d_H(f_A(p),f_A(q))\big)\lesssim \frac{1}{\sqrt{m}}.
\end{equation}
\item If $I\subset [n]$ is chosen uniformly at random from all subsets of size $m$, then for any $p,q\in \Sp^{n-1}$
         \begin{equation}
         \label{eqn:varboundMainRandom}
         \Var \big(d_H(f_A(p),f_A(q))\big)\lesssim \frac{1}{m} + \frac{1}{\sqrt{n}}.
         \end{equation}
\end{enumerate}
\end{thm}
In \cite[Theorem 3]{YBK15} it was shown that if $I=[m]$, then 
$$\Var \big(d_H(f_A(p),f_A(q))\big)\lesssim \frac{1}{m} + \rho$$ 
where $\rho=\max\{\n{p}_{\infty},  \n{q}_{\infty} \}$. This bound is non-trivial if both $p$ and $q$ are well-spread. Note that $\frac{1}{\sqrt{n}}\leq \rho$, and equality is achieved if and only
if $p$ and $q$ are perfectly spread, i.e. $p_i=q_i=\pm \frac{1}{\sqrt{n}}$ for all $i\in [n]$. Our variance bound in (2) is always better and shows that a well-spreadness assumption is not necessary if one subsamples uniformly at random. As was discussed at the beginning of this section the estimate (\ref{eqn:varboundMainRandom}) is optimal for $m\leq \sqrt{n}$. We conjecture that it is possible to remove the $1/\sqrt{n}$ factor altogether.  
\begin{proof}
We start by proving $(1)$. Since $\E X_i =d_{\Sp^{n-1}}(p,q)$, it follows that $\Var X_i =d_{\Sp^{n-1}}(p,q)-d_{\Sp^{n-1}}(p,q)^2$ and therefore (see (\ref{eqn:varExpand})) 
\begin{align*}
\Var \big(d_H(f_A(p),f_A(q))\big)=\frac{d_{\Sp^{n-1}}(p,q)- d_{\Sp^{n-1}}(p,q)^2}{m} + \frac{2}{m^2} \sum_{1\leq i<j \leq m} \Cov{X_i}{X_j}.\\
\end{align*}
By Lemma~\ref{radcov} for $i<j$ with $j-i\neq \frac{n}{2}$
\begin{align*}
|\Cov{X_i}{X_j}| \lesssim  \eu{p\od T^{j-i}p}+ \eu{p\od T^{j-i}q}+ \eu{q\od T^{j-i}p}+ \eu{q\od T^{j-i}q}.
\end{align*}
If $j-i =\frac{n}{2}$ we use the trivial bound $|\Cov{X_i}{X_j}|\leq 1$. Combining these inequalities with
$$|\{1\leq i<j \leq m\;: \; j-i=k \}|=m-k,$$
we find
\begin{align}\label{covsum}
&\sum_{1\leq i<j \leq m} |\Cov{X_i}{X_j}| \\
& \ \ =\sum_{k=1}^{m-1} \sum_{\{i<j\;:\; j-i=k\}}  |\Cov{X_i}{X_j}| \nonumber\\
& \ \ \lesssim \sum_{k=1}^{m-1} (m-k)  \big( \eu{p\od T^{k}p}+ \eu{p\od T^{k}q}+ \eu{q\od T^{k}p}+ \eu{q\od T^{k}q}\big) \nonumber\\
& \ \ \qquad +  \Big(m- \frac{n}{2}\Big) 1_{\{m-1\geq \frac{n}{2}\}} \nonumber\\
& \ \ \leq  m \sum_{k=1}^{m}   \big( \eu{p\od T^{k}p}+ \eu{p\od T^{k}q}+ \eu{q\od T^{k}p}+ \eu{q\od T^{k}q}\big) + m \nonumber\\
& \ \ \lesssim m^{3/2}\nonumber.
\end{align}
For the last inequality we have used Cauchy-Schwarz and 
$$\sum_{k=1}^n \eu{p\od T^kq}^2=\sum_{k=1}^n\sum_{i=1}^n p_i^2 q_{i+k}^2=\sum_{i=1}^n p_i^2\sum_{k=1}^n q_{i+k}^2=1\; \text{ for all } p,q\in \Sp^{n-1}.$$
We now prove (2). Let $I$ be chosen uniformly at random and let $\theta_i\in \{0,1\}$ be the induced selector variables, i.e., $\theta_i=1$ if and only if $i\in I$. Note that
\begin{align*}
 \Var \big(d_H(f_A(p),f_A(q))\big)&=\Var \Big( \frac{1}{m}\sum_{i\in I} 1_{\sgn(\skp{a_i}{p}) \neq \sgn(\skp{a_i}{q})}  \Big) \\
   							&=\frac{1}{m^2} \sum_{i=1}^n \Var(\theta_i X_i) + \frac{2}{m^2}  \sum_{1\leq i<j \leq n} \Cov{\theta_i X_i}{\theta_j X_j}.
\end{align*}							
Clearly, $\Var(\theta_iX_i)= \frac{m}{n}d_{\Sp^{n-1}}(p,q)-\frac{m^2}{n^2} d_{\Sp^{n-1}}(p,q)^2$ and for $i\neq j$, 
\begin{align*}
\Cov{\theta_i X_i}{\theta_j X_j} &= \E(\theta_i \theta_j) \E (X_iX_j)- \frac{m^2}{n^2} \E X_i \E X_j\\
		       &\leq  \frac{m^2}{n^2} \E (X_iX_j)- \frac{m^2}{n^2} \E X_i \E X_j= \frac{m^2}{n^2} \Cov{X_i}{X_j},
\end{align*}
as $\theta_i$ and $\theta_j$ are negatively correlated. Combining these estimates, we find
$$\Var \big(d_H(f_A(p),f_A(q))\big) \leq \frac{d_{\Sp^{n-1}}(p,q)}{m} - \frac{d_{\Sp^{n-1}}(p,q)^2}{n} + \frac{2}{n^2} \sum_{1\leq i<j \leq n} \Cov{X_i}{X_j}.$$
Since \eqref{covsum} holds for any $m\leq n$, we obtain
$$ \sum_{1\leq i<j \leq n} |\Cov{X_i}{X_j}|\lesssim n^{3/2}.$$
\end{proof}

\section{Fast binary embeddings}
\label{sec:binaryEmbedding}

In the following, we are interested to construct a $\del$-binary embedding for an arbitrary finite dataset $\mathcal{D}\subset \R^n$, which achieves both optimal bit complexity $\del^{-2}\log(|\mathcal{D}|)$
and runs in almost linear time $n\log(n)$. We consider modifications of the two binary embeddings from \cite{YCP15}, which were discussed in the introduction.\par 
A common step in all the binary embeddings we consider is to first reduce the dimensionality of the data set, while approximately preserving the geodesic distances between the points. The following observation, from the proof of \cite[Lemma 3.6]{YCP15}, says that this can be achieved using a map that preserves the norms of and Euclidean distances between the data points up to a multiplicative error. Let us say that a matrix $A\in \R^{n'\ti n}$ is a \emph{$\del$-isometry} on a set $\cP\subset \R^n$ if 
\begin{equation}
\label{eqn:EucIsom}
(1-\del)\|z\|_2\leq \|Az\|_2\leq (1+\del)\|z\|_2 \qquad \operatorname{for \ all} \ z\in \cP,
\end{equation}
that is, $A$ preserves norms of vectors up to a multiplicative error $\del$. 
\begin{lemma}
\label{lem:JLgeodesic}
\cite{YCP15} Let $\cD\subset \Sp^{n-1}$ be symmetric and suppose that $A$ is a $\del$-isometry on both $\cD$ and $\cD-\cD$. Then,
$$|d_{\Sp^{n'-1}}(Ax,Ay) - d_{\Sp^{n-1}}(x,y)| \leq \del \qquad \operatorname{for \ all} \ x,y\in \cD.$$
\end{lemma}
A substantial amount of research has been devoted to the construction of random matrices that satisfy (\ref{eqn:EucIsom}) with high probability. We will consider two families of constructions that allow for fast matrix-vector multiplication. The family of \emph{fast Johnson-Lindenstrauss transforms} \cite{AiC09, AiL13,KrW11,CGV13} rely on the fast Fourier transform (FFT) for fast multiplication. One particular construction from this family is the $n'\ti n$ random matrix
$$\Phi_{\operatorname{FJL}} = \sqrt{\frac{n}{n'}} R_I H D_{\eps},$$
where $I$ is a subset of $n'$ indices selected uniformly at random from $[n]$, $\eps$ is a Rademacher vector independent of $I$ and $H$ is the Hadamard transform. Using the FFT, $\Phi_{\operatorname{FJL}}x$ can be computed in time $\cO(n\log n)$. By combining the results of \cite{KrW11} and \cite{CGV13} it follows that $\Phi_{\operatorname{FJL}}$ satisfies the conditions of Lemma~\ref{lem:JLgeodesic} on a set $\cD$ of $N$ points in $\Sp^{n-1}$ with probability at least $1-\eta$ if 
$$n'\gtrsim  \del^{-2} \log(N/\eta) (\log^3(\log(N/\eta))\log(n) + \log(1/\eta)).$$ 
A different family of `fast' constructions can be obtained by sparsifying a random sign matrix. To be more precise, let $\si_{ij}$, $i\in [n']$, $j\in [n]$ be independent Rademacher random variables. We consider $\{0,1\}$-valued random variables $\del_{ij}$, which are independent of the $\si_{ij}$, with the following properties:
\begin{itemize}
\item For a fixed column $j$ the $\del_{ij}$ are negatively correlated, i.e. 
\begin{equation*}
\forall 1\leq i_1<i_2<\ldots<i_k\leq n',\  \E\Big(\prod_{t=1}^k\del_{i_t,j}\Big)\leq \prod_{t=1}^k\E \del_{i_t,j} = \Big(\frac{s}{n'}\Big)^k;
\end{equation*}
\item For any fixed column $j$ there are exactly $s$ nonzero $\del_{ij}$, i.e., $\sum_{i=1}^{n'} \del_{ij}=s$;
\item The vectors $(\delta_{ij})_{i=1}^{n'}$ are independent across different columns $1\le j\le n$.
\end{itemize}
The $n'\ti n$ \emph{sparse Johnson-Lindenstrauss transform $\Phi_{\operatorname{SJL}}$ with column sparsity $s$} \cite{DKS10,BOR10,KaN14} is defined by 
\begin{equation*}
(\Phi_{\operatorname{SJL}})_{ij}=\frac{1}{\sqrt{s}}\si_{ij}\del_{ij}.
\end{equation*}
One possible concrete implementation \cite{KaN14} is to take the columns independent, and in each column we choose exactly $s$ locations uniformly at random, without replacement, to specify the $\delta_{ij}$. For any $x\in \R^n$ with $\|x\|_0$ non-zero entries, $\Phi_{\operatorname{SJL}}x$ can be computed in time $\cO(s\|x\|_0)$. It follows from \cite{KaN14} that $\Phi_{\operatorname{SJL}}$ satisfies the conditions of Lemma~\ref{lem:JLgeodesic} on a set $\cD$ of $N$ points in $\Sp^{n-1}$ with probability at least $1-\eta$ if 
$$n'\gtrsim \del^{-2} \log(N/\eta), \qquad s\gtrsim \del^{-1} \log(N/\eta).$$
The following result under condition (i) was already obtained in \cite[Algorithm 3]{YCP15}. Note that both embeddings in Proposition~\ref{alternative} achieve the optimal bit complexity.
\begin{prop}[Accelerated Gaussian binary embeddings]
\label{alternative}
Let $\mathcal{D}=\{x_1,...,x_N\}\subset \Sp^{n-1}$. Let $G\in \R^{m\times n'}$ be a standard Gaussian matrix, i.e., its entries are independent standard Gaussians. Set
$$m\gtrsim \del^{-2}\log(N/\eta).$$ 
Suppose that one of the two conditions hold:
\begin{enumerate}
\item[(i)] $\Phi=\Phi_{\operatorname{FJL}}$ is an $n'\ti n$ FJLT with
$$n'\gtrsim  \del^{-2} \log(N/\eta) (\log^3(\log(N/\eta))\log(n) + \log(1/\eta)).$$
\item[(ii)]  $\Phi=\Phi_{\operatorname{SJL}}$ is an $n'\ti n$ SJLT with
$$n'\gtrsim \del^{-2} \log(N/\eta), \qquad s\gtrsim \del^{-1} \log(N/\eta).$$
\end{enumerate}
Set $A=G \Phi$. Then, with probability at least $1-\eta$, $f_A$ is a $\del$-binary embedding of $\mathcal{D}$ into $(\{-1,1\}^m, d_H)$, i.e.,
\begin{equation*}
\sup_{i,j\in [N]}|d_{H}(f_A(x_i), f_A(x_j)) - d_{\Sp^{n-1}}(x_i,x_j)|\leq \del.
\end{equation*} 
\end{prop} 
\begin{proof}
For $i\in [N]$ we define $y_i=\Phi x_i\in \R^{n'}$. As has been discussed before, under both (i) and (ii) we have with probability $1- \frac{\eta}{2}$
\begin{align}\label{delhadamardProp}
\sup_{i,j\in [N]} |d_{\Sp^{n'-1}}(y_i,y_j)- d_{\Sp^{n-1}}(x_i, x_j)|\leq \del/2.
\end{align}
It is well-known (see e.g.\ \cite{PlV14} or \cite{YCP15}[Proposition 2.2]) that $f_G$ is with probability $1-\frac{\eta}{2}$ a $\del/2$-binary embedding on $\{y_1,\ldots,y_N\}$, i.e.,
\begin{equation*}
\sup_{i,j\in [N]}|d_{H}(f_G(x_i), f_G(x_j)) - d_{\Sp^{n-1}}(x_i,x_j)|\leq \del/2.
\end{equation*} 
Combining these observations yields the result.
\end{proof}
As was already noted in \cite{YCP15}, for any $x\in \Sp^{n-1}$ one can compute $f_A(x)$ in time $\cO(n \log(n))$ under condition (i) if
\begin{align}\label{alternativerun}
\log(N/\eta)\lesssim \del^2 \sqrt{n}. 
\end{align}
Under condition (ii) one can compute $f_A(x)$ in time
$$\cO(\del^{-4}\log^2(N/\eta) + \del^{-1}\log(N/\eta)\|x\|_0).$$
In particular, under (\ref{alternativerun}) the embedding $f_A$ runs in linear time on $\sqrt{n}$-sparse vectors. The running time can even improve further if one is embedding a small set of very sparse vectors, say.\par
We next consider two binary embeddings, for which one $\del$-factor in \eqref{alternativerun} can be removed, provided that $n$ is large enough. Let $g^{(1)},...,g^{(B)}\in \R^{n'}$ be independent standard Gaussian random vectors and $\eps^{(1)}, ..., \eps^{(B)}\in \{-1,1\}^{n'}$ independent Rademacher vectors.
Let $I^{(1)},..., I^{(B)}\subset [n']$ be independent, uniformly random subsets of size $m'=\frac{m}{B}$.
For $s\in [B]$ set
$$\Psi^{(s)}=R_{I^{(s)}} C_{g^{(s)}} D_{\eps^{(s)}}\in \R^{m'\times n'}.$$
We now stack these matrices to obtain the $m\ti n'$ matrix
$$\Psi=\begin{pmatrix}
\Psi^{(1)}\\
 \vdots\\
\Psi^{(B)}
\end{pmatrix}
.$$
\begin{thm}[Median fast binary embeddings]\label{binaryembedding1} Let $\mathcal{D}=\{x_1,...,x_N\}\subset \Sp^{n-1}$. Suppose that one of the following two conditions hold:
\begin{enumerate}
\item[(i)] $\Phi=\Phi_{\operatorname{FJL}}$ is an $n'\ti n$ FJLT with
$$n'\gtrsim  \del^{-2} \log(N/\eta) (\log^3(\log(N/\eta))\log(n) + \log(1/\eta)).$$
\item[(ii)]  $\Phi=\Phi_{\operatorname{SJL}}$ is an $n'\ti n$ SJLT with
$$n'\gtrsim \del^{-2} \log(N/\eta), \qquad s\gtrsim \del^{-1} \log(N/\eta).$$
\end{enumerate}
Suppose, moreover, that
\begin{align*}
B\gtrsim \log(N/\eta), \qquad n' \geq \frac{m}{B}\gtrsim \del^{-2}, \qquad n' \gtrsim \del^{-4}.
\end{align*}
Set $A=\Psi\Phi$. Then with probability at least $1- \eta$, $f_A$ is a $\del$-binary embedding of $\mathcal{D}$ into $(\{-1,1\}^m, d_{\operatorname{med},B})$, i.e.,
\begin{equation*}
\sup_{i,j\in [N]}|d_{\operatorname{med},B}(f_A(x_i), f_A(x_j)) - d_{\Sp^{n-1}}(x_i,x_j)|\leq \del.
\end{equation*}
\end{thm}
Theorem~\ref{binaryembedding1} under condition (i) is essentially the result claimed in \cite[Theorem 3.8]{YCP15}. The proof in \cite{YCP15}, however, contains a gap (see Remark~\ref{rem:proofGap} below). Let us note that our construction of the $\Psi^{(s)}$ is slightly different from the one in \cite{YCP15}. Instead of a Gaussian Toeplitz matrices we use Gaussian circulant matrices (although our proof works for Toeplitz matrices as well) and, more importantly, we use uniform random subsampling instead of deterministic subsampling, in order to invoke Theorem~\ref{varbound}.    
\begin{proof} 
For $i\in [N]$ we define $y_i=\Phi x_i\in \R^{n'}$. By the discussion prior to Proposition~\ref{alternative}, under both (i) and (ii) we have with probability $1- \frac{\eta}{2}$
\begin{align}\label{delhadamard}
\sup_{i,j\in [N]} |d_{\Sp^{n-1}}(x_i, x_j)-d_{\Sp^{n'-1}}(y_i,y_j)|\leq \del.
\end{align}
Fix $i,j\in [N]$. By subsequently applying (\ref{eqn:unbEst}), Markov's inequality and (\ref{eqn:varboundMainRandom}) in Theorem~\ref{varbound}, we obtain
\begin{align*}
&\Pb( |d_H(\sgn(\Psi^{(s)}y_i),  \sgn(\Psi^{(s)}y_j))- d_{\Sp^{n'-1}}(y_i,y_j)|\geq \del )\\
& \qquad = \Pb( |d_H(\sgn(\Psi^{(s)}y_i),  \sgn(\Psi^{(s)}y_j))- \E d_H(\sgn(\Psi^{(s)}y_i),  \sgn(\Psi^{(s)}y_j))|\geq \del )\\
& \qquad \leq \del^{-2} \Var\big(d_H(\sgn(\Psi^{(s)}y_i),  \sgn(\Psi^{(s)}y_j)\big)\\
&\qquad \lesssim \del^{-2}\Big( \frac{1}{m'}+ \frac{1}{\sqrt{n'}}\Big)\leq \frac{1}{4},
\end{align*}
where in the last step we use $m'\gtrsim \del^{-2}$ and $n'\gtrsim \del^{-4}$. Defining $$E_s=1_{\{|d_H(\sgn(\Psi^{(s)}y_i),  \sgn(\Psi^{(s)}y_j))- d_{\Sp^{n'-1}}(y_i,y_j)|\geq \del \}},$$
this translates to $\E E_s\leq \frac{1}{4}$. Observe
\begin{align*}
\Big\{\sum_{s=1}^B E_s < \frac{B}{2}\Big\}\subset \Big\{ \big|d_{\operatorname{med},B}(\sgn(\Psi y_i),  \sgn(\Psi y_j))- d_{\Sp^{n'-1}}(y_i,y_j)\big|< \del \Big\},
\end{align*}
which implies using $-\E E_s \geq -\frac{1}{4}$
\begin{align*}
&\Pb( \big|d_{\operatorname{med},B}(\sgn(\Psi y_i),  \sgn(\Psi y_j))- d_{\Sp^{n'-1}}(y_i,y_j)\big|\geq  \del  )\\
&\qquad \leq \Pb\Big( \sum_{s=1}^B E_s \geq \frac{B}{2}\Big)= \Pb\Big(\frac{1}{B} \sum_{s=1}^B E_s \geq \frac{1}{2}\Big)\\
&\qquad \leq \Pb\Big(\frac{1}{B} \sum_{s=1}^B (E_s- \E E_s )\geq \frac{1}{4}\Big).
\end{align*}
Since $\{E_s\}_{s\in [B]}$ are independent, Hoeffding's inequality yields
\begin{align*}
\Pb( \big|d_{\operatorname{med},B}(\sgn(\Psi y_i),  \sgn(\Psi y_j))- d_{\Sp^{n'-1}}(y_i,y_j)\big|\geq  \del  )\leq e^{-\frac{B}{32}}.
\end{align*}
Since this holds for any $i,j\in[N]$, a union bound now implies
\begin{align}\label{delhamming}
\Pb\Big(\sup_{i,j\in [N]}\big|d_{\operatorname{med},B}(\sgn(\Psi y_i),  \sgn(\Psi y_j))- d_{\Sp^{n'-1}}(y_i,y_j)\big|\geq  \del\Big)\leq N^2 e^{-\frac{B}{32}}\leq \frac{\eta}{2},
\end{align}
where in the last step we use $B\gtrsim \log(N/\eta)$.
The triangle inequality using \eqref{delhadamard} and \eqref{delhamming}  yields
$$\sup_{i,j\in [N]}|d_{\operatorname{med},B}(f_A(x_i), f_A(x_j)) - d_{\Sp^{n-1}}(x_i,x_j)|\leq 2\del,$$
with probability at least $1- \eta$. A rescaling in $\del$ yields the result. 
\end{proof}
\begin{rem}
\label{rem:proofGap}
Let us now briefly discuss the subtle proof gap occurring in \cite{YCP15}. As was noted before, they considered mappings $\Psi^{(s)}$ which used deterministic subsampling (instead of uniform random subsampling) and Gaussian Toeplitz matrices (instead of Gaussian circulant matrices). In their proof they claimed that
\begin{align*}
\Var\Big(d_H\big(\sgn(\Psi^{(s)} y_i) , \sgn(\Psi^{(s)} y_j)\big) \Big) \lesssim \frac{1}{m'}
\end{align*}
by reasoning as follows. Let $a_k$ be the $k$-th row of $\Psi^{(s)}$. They first showed (correctly) that $\sgn(\skp{a_k}{y_i})$ and $\sgn(\skp{a_\ell}{y_j})$ are pairwise independent
for any $y_i,y_j \in \R^{n'}$ and $k\neq \ell$ \cite[Lemma 3.7]{YCP15}. They then wrote
\begin{align*}
\Var\Big(d_H\big(\sgn(\Psi^{(s)} y_i) , \sgn(\Psi^{(s)} y_j)\big) \Big) & = \Var\Big(\frac{1}{m'}\sum_{k=1}^{m'} 1_{\sgn(\skp{a_k}{y_i})\neq \sgn(\skp{a_k}{y_j})}\Big) \\
& =\frac{1}{(m')^2} \sum_{k=1}^{m'} \Var(1_{\sgn(\skp{a_k}{y_i})\neq \sgn(\skp{a_k}{y_j})}) \leq \frac{1}{m'}.
\end{align*}
Unfortunately, even though $\sgn(\skp{a_k}{y_i})$ and $\sgn(\skp{a_\ell}{y_j})$ are pairwise independent, it is \emph{not} true that 
$$1_{\sgn(\skp{a_k}{y_i})\neq \sgn(\skp{a_k}{y_j})} \qquad \text{and} \qquad 1_{\sgn(\skp{a_\ell}{y_i})\neq \sgn(\skp{a_\ell}{y_j})}$$
are pairwise independent for any $y_i,y_j$ and $k\neq \ell$. 
\end{rem}
\begin{rem} 
\label{rem:runtimeDis}
Let us compare the running times of the four different binary embeddings for $\eta$ being a constant. Note first that if 
\begin{equation}\label{binaryrun}
\log(N)\lesssim \del \sqrt{\frac{n}{\log(\del^{-1})}} \min \Big\{\del^3 \sqrt{\frac{n}{\log(\del^{-1})}} \log(n), 1\Big\},
\end{equation}
then the median fast binary embedding with FJLT can be computed in $\tilde{\cO}(n \log(n))$, where $\tilde{\cO}$ hides $\log\log$-factors. Indeed, computation time of the FJLT is $\cO(n \log(n))$
and matrix-vector multiplication for each of the $B$ blocks can be computed in 
$$\cO(n' \log(n'))=\tilde{\cO}\big( \del^{-2}\log(\del^{-1}) \max\{ \del^{-2}, \log(N) \log(n)\}\big).$$
Hence, total computation time is
\begin{align*}
&\cO (B n' \log(n')+ n\log(n) )\\
&=\tilde{\cO} \big(\del^{-2}\log(\del^{-1})\log(N) \max\{ \del^{-2}, \log(N) \log(n)\} + n\log(n) \big)  =\tilde{\cO} (n\log(n)),
\end{align*}
if (\ref{binaryrun}) is satisfied. In particular this holds if $n\gtrsim \del^{-6}$ and $\log(N)\lesssim \del \sqrt{n}/\sqrt{\log(1/\del)}$. In comparison, the accelerated Gaussian binary embedding with FJLT achieves running time $\cO (n\log(n))$, if $\log(N)\lesssim \del^2 \sqrt{n}$.\par 
Let us now consider the median fast binary embedding with the SJLT. The running time of this embedding is 
\begin{align*}
& \cO(Bn'\log(n') + \del^{-1}\log(N)\|x\|_0) \\
& \qquad = \tilde{\cO}(\log(\del^{-1})(\del^{-2}\log^2(N) + \del^{-4}\log(N)) + \del^{-1}\log(N)\|x\|_0).
\end{align*}
Thus, this embedding can achieve a (near-)linear embedding time in a larger range of $N$ compared to the accelerated Gaussian binary embedding with SJLT. Indeed, if $\log N\leq \del\sqrt{n}$ and $n\gtrsim \del^{-6}$, then the running time is $\tilde{O}(n\log n+\sqrt{n}\|x\|_0)$ which is near-linear if $\|x\|_0 \leq \sqrt{n}$. In comparison, for the accelerated Gaussian binary embedding we required $\log N\leq \del^2\sqrt{n}$ to obtain linear running time on $\sqrt{n}$-sparse vectors.\par
In contrast to the running times for the embeddings involving the FJLT, the running times for the SJLT-based embeddings can be even faster than $\cO(n\log n)$, for instance if the number of vectors $N$ is small and the vectors are sparse.

Let us finally note that the median fast binary embeddings with FJLT and SJLT can both be computed in time $\tilde{\cO}(n \log(n))$ if  $n\gtrsim \del^{-4}$ (instead of $n\gtrsim \del^{-6}$), if additionally $\log N\leq \del^2\sqrt{n}$ (the condition for the accelerated Gaussian binary embeddings) holds.
\end{rem}
\begin{rem}
In the proof of Theorem~\ref{binaryembedding1}, the optimal $1/m$ scaling in (\ref{eqn:varboundMainRandom}) is essential. In particular, it is not possible to instead use the variance bound in (\ref{eqn:varboundMainDet}). It would be interesting to remove the $1/\sqrt{n}$ factor in (\ref{eqn:varboundMainRandom}). This would remove the condition $n\gtrsim \del^{-6}$ from the discussion in Remark~\ref{rem:runtimeDis}. 
\end{rem}

\bibliographystyle{amsplain}
\bibliography{circBib}

\end{document}